\documentclass{kybernetika}
\usepackage{amsthm}
\newcommand\twoline[2]{\genfrac{\{}{\}}{0pt}{}{\strut#1}{\strut#2}}
\newcommand\R{\mathbf R}
\DeclareMathAlphabet{\mathsfsl}{OML}{cmm}{b}{it}
\renewcommand\r{\mathsfsl{r}}
\newcommand\restr{\mathord\upharpoonright}
\newcommand\down{{\downarrow}}
\newtheorem{theorem}{Theorem}
\newtheorem{lemma}[theorem]{Lemma}
\newtheorem{proposition}[theorem]{Proposition}
\newtheorem{corollary}[theorem]{Corollary}
\newtheorem{example}[theorem]{Example}

\newtheorem{problem}{Problem}
\newtheorem{polymatroidconjecture}{Sticky polymatroid conjecture}

\def\eps{\varepsilon}

\def\H{\mathbf{H}}
\hidedoi
\setheader{K\,Y\,B\,E\,R\,N\,E\,T\,I\,K\,A\,
}

\begin{document}
\pagestyle{myheadings}
\title{\bf One-adhesive polymatroids}
\author{Laszlo Csirmaz}
\contact{Laszlo}{Csirmaz}{Central European University}{csirmaz@renyi.hu}
\date{\small\it Dedicated to the memory of Frantisek Mat\'u\v s}
\markboth{L.~Csirmaz}{One-adhesive polymatroids}

\maketitle

\par\vspace*{-18mm}\emph{Dedicated to the memory of Frantisek Mat\'u\v s}

\vspace{5mm}

\begin{abstract}
Adhesive polymatroids were defined by F.~Mat\'u\v s motivated by entropy
functions. Two polymatroids are adhesive if they can be glued together along
their joint part in a modular way; and are one-adhesive, if one of them has
a single point outside their intersection. It is shown that two polymatroids
are one-adhesive if and only if two closely related polymatroids have joint
extension. Using this result, adhesive polymatroid pairs on a five-element
set are characterized.
\end{abstract}
\keywords{polymatroid, amalgam, adhesive polymatroid,
entropy function, polyhedral cone}

\classification{05B35, 94A15, 52B12}

\section{Preliminaries}\label{sec:intro}

A {\em polymatroid $(f,M)$} is a non-negative, monotone and submodular
function $f$ defined on the collection of non-empty subsets of the finite
set $M$. Here $M$ is the {\em ground set}, and $f$ is the {\em rank function}.
The polymatroid is {\em integer} if all ranks are integer. An integer
polymatroid is a {\em matroid}, if the rank of singletons are either zero or
one. Matroids are combinatorial objects which generalize the properties of
linear dependence among a finite set of vectors. For an introduction to
matroids, see \cite{oxley}; and about polymatroids consult \cite{Lov,fmadhe}. The
rank function $f$ can be identified with a $(2^{|M|}-1)$-dimensional real vector,
where the indices are the non-empty subsets of $M$. 
In this paper the {\em distance} of two polymatroids $f$ and $g$ on the same ground
set is measured as the usual Euclidean distance of the corresponding vectors, and is
denoted as $\|f-g\|$.

Following the usual practice, ground sets and their subsets are denoted by
capital letters, their elements by lower case letters. The union sign
$\cup$ is frequently omitted as well as the curly brackets around singletons,
thus $Aab$ denotes the set $A\cup\{a,b\}$. For a function $f$
defined on the subsets of the finite set $M$ (such as the rank function of a
polymatroid) the usual information-theoretical abbreviations are used. Here
$I$, $J$, $K$ are disjoint subsets of the ground set:
\begin{align*}
   f(I,J|K) &= f(IK)+f(JK)-f(IJK)-f(K),\\
   f(I,J)   &= f(I,J|\emptyset) = f(I)+f(J)-f(IJ)-f(\emptyset),\\
   f(I|K)   &= f(IK)-f(K).
\end{align*}
When $f$ is a rank function, $f(\emptyset)$ is considered to be zero.
In cases when the function $f$ is clear from the context, even $f$ is
omitted. 
Additionally, the {\em Ingleton
expression} \cite{ingleton} is abbreviated as
$$
   f[I,J,K,L] = -f(I,J)+f(I,J|K)+f(I,J|L)+f(K,L).
$$
Observe that it is symmetrical for swapping $I$ and
$J$ as well as swapping $K$ and $L$. 

Vectors corresponding to polymatroids on the ground set $M$ form the pointed polyhedral cone 
$\Gamma_M$ \cite{yeungbook}. 
Its facets are the hyperplanes determined by the
basic submodular inequalities $(i,j|K)\ge 0$ with distinct $i,j\in M{-}K$
and $K\subseteq M$ ($K$ can be empty), and the monotonicity
requirements $(i|M{-}i)\ge 0$, see \cite[Theorem 2]{fmadhe}. Much less is known
about the extremal rays of this cone. They have been computed for ground sets
up to five elements \cite{studeny-kocka}, without indicating any structural
property.

\subsection{Entropic, linear, and modular polymatroids}

An important class of polymatroids describes the entropy structure of the
marginals of finitely many discrete random variables. Assume
$\{\xi_i:i\in M\}$ is a collection of (jointly distributed) random
variables. For $A\subseteq M$ let $\H(\xi_A)$ be the usual Shannon entropy 
of the marginal distribution $\xi_A = \{ \xi_i: i\in A\}$. The function $f(A) =
\H(\xi_A)$ is a polymatroid \cite{fuji}. Such polymatroids are called {\em
entropic}, and the collection of entropic polymatroids is
$\Gamma^*_M\subseteq \Gamma_M$ \cite{yeungbook}. The closure of $\Gamma^*_M$ 
(in the usual
Euclidean topology) is the collection of {\em almost entropic} or
{\em aent} polymatroids.
Studying
polymatroids is motivated partly by the difficult task of understanding the
entropic region as well as solving problems arising in secret sharing
\cite{padro1,seymour}, network coding \cite{acy}, and other areas.

Another important subclass is the linear polymatroids. 
$(f,M)$ is {\em linearly representable} if there is a vector
space $V$ over some finite field, linear subspaces $V_i\subseteq V$ for
each $i\in M$, such that $f(A)$ is the dimension of the linear subspace
spanned by the vectors in $\bigcup\{ V_i:i\in A\}$. Linearly
representable polymatroids are integer.
A polymatroid is {\em
linear} if it is in the conic hull of linearly representable polymatroids,
namely, it can be written as a non-negative linear combination of such
polymatroids. Linear polymatroids are almost entropic, see
\cite{entreg,M.twocon,padro}.

The polymatroid $(f,M)$ is {\em modular} if $f(I,J)=0$ for any two
disjoint non-empty subsets $I,J\subset M$, or, equivalently, if 
for all $A\subseteq M$ we have
$$
    f(A) = {\textstyle\sum}\, \{ f(i)\,:\, i\in A\}.
$$
Modular polymatroids are entropic and linear \cite{fmadhe}.

In matroid theory modularity refers to a different notion \cite{oxley}, which
will be called \emph{flat-modularity} here. $F\subseteq M$ is a \emph{flat} if
its rank is strictly smaller than that of any of its proper extensions. The 
polymatroid $(f,M)$ is \emph{flat-modular} if every pair $(F_1,F_2)$ of its flats
forms a modular pair, namely the submodularity holds with equality:
$$
   f(F_1)+f(F_2) = f(F_1\cap F_2)+f(F_1\cup F_2).
$$
Modular polymatroids are flat-modular, but the converse is not true in
general.

\medskip
For a subset $A\subset M$ define the function $\r_A$ on (non-empty)
subsets of $M$ as follows:
$$
    \r_A(I) =\min\{1,|A\cap I|\} =  \begin{cases} 1 & \mbox{if $A\cap I$ is not empty,}\\
                            0 & \mbox{otherwise}.
              \end{cases}
$$
Clearly $(\r_A,M)$ is a matroid and linearly representable over any vector
space. They are linearly independent and span the whole
$2^{|M|-1}$-dimensional space. This is immediate from the fact that the
linear combination
\begin{equation}\label{eq:r-coeffs}
    \sum_{B\subseteq A} \,(-1)^{1+|A{-}B|}\, \r_{\!N{-}B}(I)
\end{equation}
takes one at $A$, and zero anywhere else, see \cite[Lemma 3]{M.twocon}.

\smallskip

It is well known that all polymatroids on two or three 
elements are linear, moreover a polymatroid $f$ on the four element set
$abcd$ is linear if and only if it satisfies all six instances of the
Ingleton inequality:
\begin{align*}
   f[a,b,c,d]\ge 0, ~~~ f[a,c,b,d]\ge 0, ~~~ f[a,d,b,c]\ge 0,\\
   f[b,c,a,d]\ge 0, ~~~ f[b,d,a,c]\ge 0, ~~~ f[c,d,a,b]\ge 0,
\end{align*}
see \cite{matus-studeny}.
Linear polymatroids on a five element set can also be characterized by 
some finite set of linear inequalities \cite{dougherty5}. Polymatroids on
ground set of size five or less have the following {\em simultaneous
approximation property}, which will be used in Section \ref{sec:sticky-2}.
\begin{proposition}\label{prop:pm-approx}
Let $|M|\le 5$, and
let $f_1$ and $f_2$ be linear polymatroids on $M$.
For each positive $\eps$ and large enough vector space $V$ there is a
$\lambda>0$ and integer polymatroids $\ell_1$ and $\ell_2$ on $M$
linearly representable over $V$, such that $\|f_i-\lambda\ell_i\|<\eps$,
additionally $\ell_1(I)=\ell_2(I)$ whenever $f_1(I)=f_2(I)$ ($I\subseteq
M$).
\end{proposition}
\begin{proof}
On ground set $|M|\le5$ linear polymatroids form a polyhedral cone.
Moreover, for every large enough vector space $V$, extremal rays of this
cone contain polymatroids linearly representable over $V$, see
\cite{dougherty5,matus-studeny}. Non-negative
rational combinations of these polymatroids form a dense subset of
linear polymatroids. Let $\ell_1$ and $\ell_2$ be such combinations with
$\|f_i-\ell_i\|<\eps$. The linearly representable polymatroids $\r_A$ span
the whole linear space, thus there are rational coefficients 
$\alpha_A$ such that
$$
   \sum_{A\subseteq M}\,\alpha_A\,\r_A(I) =
   \begin{cases}
      \ell_1(I)-\ell_2(I) & \mbox{ if $f_1(I)=f_2(I)$},\\
      0                   & \mbox{ otherwise. }
   \end{cases}
$$
As $|\ell_1(I)-\ell_2(I)|<2\eps$ whenever $f_1(I)=f_2(I)$,
(\ref{eq:r-coeffs}) implies that all coefficients $\alpha_A$ have absolute
value smaller than $2^{|M|+1}\eps$. Using the notation
$\alpha^+=\max\{0,\alpha\}$ and $\alpha^-=\max\{0,-\alpha\}$, the
polymatroids
\begin{align*}
  &  \ell_1+{\textstyle\sum}_{A\subseteq M}\,\alpha^+_A\,\r_A \\
  &  \ell_2+{\textstyle\sum}_{A\subseteq M}\,\alpha^-_A\,\r_A
\end{align*}
are non-negative rational combinations of linearly representable
polymatroids; are equal whenever $f_1(I)=f_2(I)$; and are approximating
$f_1$ and $f_2$, respectively, better than $2^{2|M|+2}\eps$.

Finally, integer combinations of linearly 
representable polymatroids over the same vector space $V$ are linearly
representable, which implies the claim.
\end{proof}

\subsection{Amalgam and adhesive extension}\label{subsec:adhesive}

Let $M$, $X$, and $Y$ be disjoint sets. 
Polymatroids $f_X$ and $f_Y$ on the ground sets $M\cup X$ and $M\cup Y$,
respectively, with joint restriction on $M$,
have an {\em amalgam}, or {\em can be glued together}, if there is a
polymatroid $f$ on $M\cup X\cup Y$ extending both $f_X$ and $f_Y$
\cite{oxley}. This
extension is {\em modular} if, in addition, $X$ and $Y$ are independent over
$M$, that is, $f(X,Y|M)=0$. If $f_X$ and $f_Y$ have such a modular
extension $f$, then $f_X$ and $f_Y$ are {\em adhesive}, and $f$ is an {\em
adhesive extension}. Adhesive extensions were defined and studied by
F.~Mat\'u\v s in \cite{fmadhe}. The main observation is that restrictions of
an almost entropic polymatroid are adhesive \cite[Lemma 2]{fmadhe}. In this
paper we investigate adhesive extensions on their own right.

When speaking about amalgam, or adhesive extension, the polymatroids are
tacitly assumed to have the same restriction on the intersection of their
ground sets.

\smallskip

We have defined the amalgam of $f_X$ and $f_Y$ as a {\em polymatroid} extending
both $f_X$ and $f_Y$. The amalgam of two matroids is traditionally required to be a
matroid. It is an interesting problem to decide whether the two different
notions of amalgam coincide.
\begin{problem}\label{problem:1}\rm
Suppose the matroids $f_X$ and $f_Y$ on $M\cup X$ and $M\cup Y$,
respectively, have a polymatroid amalgam on $M\cup X\cup Y$.
Is it true that then they have a matroid amalgam as well?
\end{problem}
\noindent
If the joint extension is integer valued then it must be a matroid; and if
there is a joint extension at all, then there is one with rational values.

\smallskip

Whether two matroids have an amalgam is a combinatorial question; the
same question about polymatroids is a {\em geometrical} one. Polymatroids
$f_X$ and $f_Y$ have an amalgam if and only if the point $(f_X,f_Y)$ (merged
along coordinates corresponding to subsets of $M$) is in the {\em
coordinatewise projection} of the polymatroid cone $\Gamma_{MXY}$ to the
subspace with coordinates $I\subseteq MXY$ where $I\subseteq MX$ or
$I\subseteq MY$.
The projection is a polyhedral cone whose bounding
hyperplanes correspond to (homogeneous) linear inequalities on the projected
coordinates. Thus $f_X$ and $f_Y$ have an amalgam if and only if the
vector $(f_X,f_Y)$ satisfies all of these inequalities. While
theoretically simple, in practice it is unclear how to calculate
the facets of the projection efficiently.

The same reasoning applies to adhesive extension. Such an extension
satisfies the additional constraint $f(X,Y|M)=0$, thus the modular extensions
form a subcone of dimension one less: the intersection of
$\Gamma_{MXY}$ and the hyperplane $f(XM)+f(YM)-f(XYM)-f(Y)=0$. $f_X$ and
$f_Y$ have an adhesive extension if an only if the point $(f_X,f_Y)$ is in
the projection of this restricted cone.

\smallskip
The polymatroid $h$ is {\em sticky} if any two extensions
of $h$ have an amalgam. Flat-modular polymatroids are sticky, the proof in
\cite[Theorem 12.4.10]{oxley} works in the polymatroid case as well, but see
also \cite[Theorem 1]{fmadhe}. The ``sticky matroid conjecture'' asserts that all
sticky matroids are flat-modular \cite{bonin}. The same conjecture is stated here
for polymatroids.
\begin{polymatroidconjecture}
Sticky polymatroids are flat-modular.
\end{polymatroidconjecture}
Factors of sticky polymatroids are sticky, and the collection of sticky
polymatroids on a given ground set forms a closed cone, thus to settle the
above conjecture it is enough to consider polymatroidal extensions of a
matroid. Consequently, if the answer to Problem \ref{problem:1} is {\em yes}
and the sticky matroid conjecture is true, then so is the sticky polymatroid
conjecture.

\medskip

To state some of our results we need one more definition.
The polymatroid $(h,M)$ is {\em $k$-$\ell$-sticky}, if any two of its extensions
$(f_X,MX)$ and $(f_Y,MY)$ with $|X|\le k$ and $|Y|\le\ell$ have an amalgam.
A polymatroid is {\em $k$-sticky}, if it is $k$-$k$-sticky.
Sticky polymatroids on small ground sets are discussed in Sections
\ref{sec:sticky-2} and \ref{sec:sticky-3}.

\subsection{New polymatroids from old ones}\label{subsec:constructions}

Each polymatroid can be decomposed as a sum of a modular and a tight
polymatroid as described in Lemmas \ref{lemma:down} and 
\ref{lemma:tightening}; it is
a generalization of \cite[Lemma 2]{entreg}. Lemma \ref{lemma:excess}
discusses how one can extend a polymatroid adding a new element to the base set.
The method will be used in later sections to create several extensions.
Recall that $\r_A$ is the polymatroid defined by $\r_A(I)=\min\{1,|A\cap
I|\}$.
\begin{lemma}\label{lemma:down}
Let $(f,M)$ be a polymatroid and $A\subset M$. Suppose the real number $\lambda$ satisfies the
following conditions:
$$\begin{array}{r@{\;\le\;}ll}
\lambda & f(x,y|B)    &\mbox{for {\em different} $x,y\in A$ and all $B\subseteq
M{-}A$; and}\\[3pt]
\lambda & f(x|M{-}A)  &\mbox{for every $x\in A$.}
\end{array}$$
Then $(f-\lambda \r_A,M)$ is a polymatroid.
\end{lemma}
\noindent
Observe that if $A$ has a single member $a$, then the first condition
holds vacuously, and the second condition simplifies to $\lambda \le
f(a|M{-}a)$.
\begin{proof}
The claim clearly holds when $\lambda\le 0$, so assume $\lambda>0$, and
let $f^*=f-\lambda \r_A$. If $I$ and $A$ are disjoint, then $f(I)-f^*(I)=0$,
in the other cases this difference is $\lambda$. One has to check the monotonicity
for the special case $f^*(Cx)-f^*(C)\ge 0$, $Cx\subseteq M$ only. This difference equals to
$f(Cx)-f(C)$ except when $A$ and $C$ are disjoint and $x\in A$. But then
\begin{align*}
   f^*(Cx)-f^*(C)&=f(Cx)-f(C)-\lambda \\
               {}&=f(x|C)-\lambda \ge f(x|M{-}A)-\lambda\ge 0
\end{align*}
by assumption.

To check submodularity, observe that $f^*(x,y|B)=f(x,y|B)$ except when $A$
and $B$ are disjoint and both $x$ and $y$ are in $A$. In the latter case
$f^*(x,y|B)=f(x,y|B)-\lambda$, which is non-negative by the fist assumption.
\end{proof}

Let $(f,M)$ be any polymatroid and $a\in M$. By the remark above,
$f-\lambda\r_a$ is a polymatroid whenever $\lambda\le f(a|M{-}a)$. 
Choosing $\lambda$ to be this maximal value,
the polymatroid $f-\lambda\r_a$ is denoted by
$f\down a$, and called {\em tightening of $f$ at} (or on) $a$. $f$ is
{\em tight at $a$} if $f=f\down a$, that is, if $f(a|M{-}a)=0$.
Note that
$(f\down a)\down a = f\down a$, thus $f\down a$ is tight at $a$; moreover 
$(f\down a)\down b = (f\down
b)\down a$. Thus one can define the {\em tight part of $f$} at
$A=\{a_1,\dots,a_k\}$ as $f\down a_1\down\cdots\down a_k$. $f$ is {\em tight
on $A$}, if $f=f\down A$, and is {\em tight} if $f=f\down M$. The next
lemma summarizes the properties of tightening used in this paper, see
\cite{entreg}.

\begin{lemma}\label{lemma:tightening}
Let $(f,M)$ be a polymatroid and $A\subseteq M$. \begin{itemize}
\item $f$ is tight on $A$ if and only if it is tight on all elements of $A$.
\item $f\down A$ is tight on $A$.
\item $f-f\down A = \sum_{a\in A}\,f(a|M{-}a)\cdot\r_a$ is a modular polymatroid.
\item $f\down M$ is tight, and $f=f\down M + (f-f\down M)$ is the unique
decomposition of $f$ into the sum of a tight and modular part. \qed
\end{itemize}
\end{lemma}

\medskip

In the last part of this section we investigate how to extend the
polymatroid $(f,M)$ to the ground set $Mx$ using the {\em excess
function} $e(A)=f(xA)- f(A)$ defined for all subsets $A\subseteq M$ (including
the empty set). In agreement with the previous notation,
$e(a,b|A)$ abbreviates $e(aA)+e(bA)-e(abA)-e(A)$, in particular,
$e(a,b)=e(a,b|\emptyset)=e(a)+e(b)-e(ab)-e(\emptyset)$.

\begin{lemma}\label{lemma:excess}
Suppose $x$ is not in the ground set $M$ of the polymatroid $f$. Extend
$f$ to the subsets of $Mx$ by $f_x(Ax) = f(A)+e(A)$. Then $f_x$ is a
polymatroid on $Mx$ if and only if the following conditions hold:
\begin{enumerate}
\item[\rm 1.] $e$ is non-negative and non-increasing: $e(A)\ge e(B) \ge 0$
for $A\subseteq B\subseteq M$;
\item[\rm 2.] $e(a|M{-}a)+f(a|M{-}a)\ge 0$ for all $a\in M$;
\item[\rm 3.] $e(a,b|A)+f(a,b|A)\ge 0$ for all $abA\subseteq M$.
\end{enumerate}
\end{lemma}
\begin{proof}
An easy case by case checking.
\end{proof}
As $e$ is non-increasing, $e(A|B)\le 0$; in particular $e(a|M{-}a)\le 0$ for
all $a\in M$. On the other hand, $e(a,b|A)$ can take both positive and negative
values even for the same excess function.

\begin{example}\label{example:excess-1}
Let $f$ be a polymatroid on $M$ and $\,0\le u,t$. Define the excess function
$e_x$ by
$$
   e_x(A) = \left\{\begin{array}{ll}
         u+t & \mbox{ if $A=\emptyset$,}\\
         u & \mbox{ otherwise}.
          \end{array}\right.
$$
If $t\le f(a,b)$ for all pairs $a,b\in M$, then $f_x$ is a polymatroid .
\end{example}
\begin{proof}
Conditions 1 and 2 of Lemma \ref{lemma:excess} trivially hold. As for
Condition 3,
$e_x(a,b|A)$ is zero except when $A=\emptyset$, and then $e_x(a,b)=-t$. Thus
it also holds by the assumption on $t$.
\end{proof}
\noindent
An easy calculation shows that for this extension $f_x$, for all pairs
$a,b\in M$ and non-empty $A\subseteq M{-}a$ we have $f_x(x,a|A)=0$, 
and $f_x(a,b|x)=f(a,b)-t$.
\begin{example}\label{example:excess-2}
Let $c\in M$ and $\,0\le u,t$. Define the excess function $e_x$ by
$$
    e_x(A) =\left\{\begin{array}{ll}
         u+t & \mbox{ if $A=\emptyset$ or $A=\{c\}$,}\\
         u & \mbox{ otherwise}.
          \end{array}\right.
$$
If $t\le f(a,b)$ and $t\le f(a,b|c)$ for all pairs $a,b\in M{-}c$,
then $f_x$ is a polymatroid.
\end{example}
\begin{proof}
Similar to the previous Example. Conditions 1 and 2 hold, moreover
$e_x(a,b|A)$ is either zero or $-t$, and the latter case holds when $A=\emptyset$
and $a,b\in M{-}c$, or when $A=\{c\}$. Thus in all cases Condition 3 holds
as well.
\end{proof}


\section{Adhesivity versus amalgam}\label{sec:adh-eq-amalgam}

As defined in Section \ref{subsec:adhesive}, polymatroids $f_X$ and $f_Y$ on
ground set $MX$ and $MY,$ respectively, have an {\em amalgam} if there is a
polymatroid
on $MXY$ extending both $f_X$ and $f_Y$. The same polymatroids are
{\em adhesive} if, in addition, they have a modular extension.
When $Y$ has a single element $y$, then the polymatroid on $My$ will be
denoted by $f_y$. In this special case adhesivity of $f_X$ and $f_y$ is
equivalent to the existence of the amalgam of closely related
polymatroids. Recall that $f_y$ is {\em tight on $y$} if $f_y(y|M)=0$,
and by tightening $f_y$ on $y$ one gets the (tight) polymatroid
$$
    f_y\down y = f_y - f_y(y|M)\cdot\r_y.
$$

\begin{theorem}\label{thm:amalgam-Xy}
Polymatroids $f_X$ and $f_y$ are adhesive if and only if $f_X$ and
$f_y\down y$ have an amalgam.
\end{theorem}
\begin{proof}
First let $g$ be the modular extension of $f_X$ and $f_y$, that is
$g(X,y|M)=0$. This equality rewrites to
$$
   g(y|MX) = g(XMy)-g(MX)=g(My)-g(M)=f_y(My)-f_y(M)=f_y(y|M).
$$
Let $g^*=g\down y$. The above equality means that restricting $g^*$ to $My$
one gets $f_y\down y$, and, as $g$ and $g^*$ on $MX$ are the same, 
restricting $g^*$ to $MX$ one gets
$f_X$. Consequently $g^*$ is the required amalgam of $f_X$ and $f_y\down y$.

Conversely, let $g^*$ be an amalgam of $f_X$ and $f_y\down y$. Then using
that $f_y\down y$ is tight on $y$, 
$g^*(My)=f_y\down y(My) = f_y\down y(M) = g^*(M)$, thus
$$
  g^*(XMy)-g^*(XM)\le g^*(My)-g^*(M)=0,
$$
which means that $g^*(X,y|M)=0$. Let
$g=g^*+\lambda\r_y$ with $\lambda=f_y(y|M)$. Then $g$ extends $f_X$ (as
$g\restr MX=g^*\restr MX=f_X$), and $f_y$ (as $g\restr My=
g^*\restr My +\lambda\r_y = (f_y-\lambda\r_y)+\lambda\r_y$). Finally,
$g(X,y|M)=g^*(X,y|M)=0$, as required.
\end{proof}

The last step in the proof works in a more general setting.

\begin{proposition}\label{prop:tight-XY}
Suppose $f_X\down X$ and $f_Y\down Y$ have an amalgam. Then $f_X$ and $f_Y$
have an amalgam as well. 
\end{proposition}
\begin{proof}
If $g$ is an amalgam of $f_X\down X$ and $f_Y\down Y$, then
$g+(f_X-f_X\down X) + (f_Y-f_Y\down Y)$ is an amalgam of $f_X$ and $f_Y$.
\end{proof}

In particular, to show that $f$ is sticky, it is enough to consider
extensions $f_X$ and $f_Y$ which are tight on $X$ and $Y$, respectively. The
condition stated in Proposition \ref{prop:tight-XY} is sufficient but not 
necessary. Polymatroids $f_x$ and $f_y$ in
Example \ref{example:ex1} have an amalgam but are not adhesive. Thus, by
Theorem \ref{thm:amalgam-Xy}, $f_x\down x$ and $f_y\down y$ have no
amalgam.


\section{One-element extensions}\label{sec:sticky-2}

This section starts with an alternative proof for a result of F.~Mat\'u\v s 
\cite{fmadhe} which claims, using our terminology, that polymatroids on two
element sets are 1-sticky. A similar proof to this one will be given for
Theorem \ref{thm:nonsticky-2}.
Theorem \ref{thm:amalgam3} gives a sufficient and necessary condition for
a pair of one-element extensions of a polymatroid on three elements 
to have an amalgam. 
Using Theorem \ref{thm:amalgam-Xy}, this
is turned into sufficient and necessary conditions for such
polymatroid pairs to be adhesive, which, in turn, yields new 5-variable 
non-Shannon entropy inequalities stated in Corollary \ref{corr:5ineq}.

The section concludes with several examples. The first one specifies two
linearly representable (entropic) polymatroids which have an amalgam, but
are not adhesive. Thus there are two linearly representable polymatroids
which have a polymatroid extension, but no almost entropic (or linear)
extension. Finally, two general examples are presented for 1-sticky and not
1-sticky polymatroids on three elements.

\def\matusthm{\cite[Corollary 2]{fmadhe}}
\begin{theorem}[\matusthm]\label{thm:1-1-sticky}
All Polymatroids $f_x$ and $f_y$ on the ground sets $abx$ and $aby$ with
common restriction to $ab$
are adhesive. In particular, such polymatroids have 
an amalgam, thus every polymatroid on a two element set is 1-sticky.
\end{theorem}
\begin{proof}
As discussed in Section \ref{subsec:adhesive}, adhesive polymatroid pairs
$(f_x,f_y)$ form a polyhedral cone. Consequently, $(f_x,f_y)$ is adhesive
if and only if $(\lambda f_x,\lambda f_y)$ is adhesive for some (or all)
positive $\lambda$. The adhesive cone is closed, thus to show that a
particular pair $(f_x,f_y)$ is
adhesive, it is enough to find, for each positive $\eps$, some adhesive pair
$(\ell_x,\ell_y)$ such that $\|f_x-\lambda\ell_x\|<\eps$, and
$\|f_y-\lambda\ell_y\|<\eps$. In this particular case $\ell_1$ and $\ell_2$
will be the linearly representable polymatroids guaranteed by
Proposition \ref{prop:pm-approx}. Thus $\ell_1$ and $\ell_2$ are represented
over the same vector space $V$,
$\lambda\ell_x$ and $\lambda\ell_y$ are 
$\eps$-close to $f_x$ and $f_y$, respectively, and the linear subspaces
in both representations corresponding to
subsets of $\{ab\}$ have the same dimensions: $\ell_x(a)=\ell_y(a)$,
$\ell_x(b)=\ell_y(b)$ and $\ell_x(ab)=\ell_y(ab)$ as these
equalities are true for the polymatroids $f_x$ and $f_y$. To conclude
the claim of the theorem it is enough to show that $(\ell_x,\ell_y)$ is an
adhesive pair.

The dimensions of subspaces spanned by $V_a$, $V_b$, and $V_a\cup V_b$ are
the same in both representations. Choose a base in the first representation
which can be partitioned to $B^x_x\cup B^x_a \cup B^x_b\cup B^x_{ab}$ such
that $\ell_x(a)=|B^x_a\cup B^x_{ab}|$, $\ell_x(b)=|B^x_b\cup B^x_{ab}|$, and
$\ell_x(ab)=|B^x_a\cup B^x_b\cup B^x_{ab}|$, and similarly for $\ell_y$.
Identify $B^x_a$ and $B^y_a$, $B^x_b$ and $B^y_b$, $B^x_{ab}$ and
$B^y_{ab}$, and take the vector space with base $B^x_x\cup B_a\cup B_b\cup
B_{ab}\cup B^y_y$ (that is, glue the representations of $\ell_x$ and
$\ell_y$ along their common part). It will be a linear representation of a
polymatroid on $abxy$, where $x$ and $y$ are independent given $ab$.
Consequently $\ell_x$ and $\ell_y$ have an adhesive extension, which
concludes the proof.
\end{proof}

Now we turn to the case of one-point extensions of polymatroids on
three-element sets. If not mentioned otherwise, all polymatroids in the rest
of this section are extensions of a fixed polymatroid on $M=\{a,b,c\}$.

\begin{theorem}\label{thm:amalgam3}
Polymatroids $f_x$, $f_y$ on the ground sets $abcx$ and $abcy$ have an amalgam if 
and only if the following eight inequalities and their permutations
(permuting $a,b,c$ and $x,y$) hold, where either the top or the bottom expression
is chosen from all three pairs in curly brackets: 
\begin{align}\label{eq:amalgam3}
  & ~~~ f_x(a,x|c)+f_x(a,b|x)+f_y(a,b|y)+f_y(c,y)+{}\\
  &+\twoline{f_x(b,x|ac)}{f_y(b,y|ac)} + \twoline{f_x(c,x|ab)}{f_y(c,y|ab)} + 
    2\twoline{f_x(x|abc)}{f_y(y|abc)} \ge f_x(a,b).\nonumber
\end{align}
\end{theorem}
\begin{proof}
It is clear that all terms are defined over one of
the polymatroids $f_x$ and $f_y$. Also, these inequalities hold for
any polymatroid with ground set $abcxy$, which can easily be checked
using an automated entropy checker, thus they {\em must} hold when $f_x$ and
$f_y$ have an amalgam. Actually, inequalities in (\ref{eq:amalgam3}) written
in basic terms and rearranged, are equivalent to
\begin{align}
   & ~~~ (a,b|xy)+(x,y|a)+(x,y|b)+(c,y|x)+(a,x|cy)+{}\label{eq:amalgamproof}\\
   &+\twoline{(b,x|acy)}{(b,y|acx)}+\twoline{(c,x|aby)}{(c,y|abx)} +
    2\twoline{(x|abcy)}{(y|abcx)} \ge 0,\nonumber
\end{align}
which evidently holds for any polymatroid on five elements.

The sufficiency can be checked by the method indicated in Section
\ref{subsec:adhesive}. Polymatroids $f_x$ and $f_y$ determine 23 out of the
31 coordinates of the polymatroid on $N=abcxy$. The missing 
8 variables are indexed by subsets of the form $Axy$ with $A\subseteq
\{abc\}$.

The facets of the polymatroid cone $\Gamma_N$ are determined by the basic submodular
inequalities $(i,j|K)\ge 0$ and by the monotonicity requirements $(i|N{-}i)\ge
0$. The strong duality of linear programming says that the facet equations of the
projection are non-negative linear combinations of these inequalities in
which the combined coefficients of the projected (dropped) variables are zero. Let $\mathcal
M$ denote the matrix whose columns are indexed by the non-empty subsets of
$N$, and whose rows contain the coefficients of the bounding facets of
$\Gamma_N$ as discussed above. In each row there are two, three, or four non-zero entries only.
When $\mathcal M$ is restricted to the eight columns labeled by $xyA$,
27 different non-zero rows remain. Let $\mathcal M'$ be this 27 by 8 matrix.
\begin{table}[!bhtb]
\begin{center}\begin{tabular}{|cccccccc|ll|}
\hline
\rule{0pt}{11pt}$xy$&$axy$&$bxy$&$cxy$&$abxy$&$acxy$&$bcxy$&$abcxy$&&\\
\hline
\rule{0pt}{12pt}
-1 & 0 & 0 & 0 & 0 & 0 & 0 & 0 & $(x,y)$&\\[2pt]
1 & 0 & 0 & -1 & 0 & 0 & 0 & 0 & $(c,x|y)$&  $(c,y|x)$\\[2pt]
-1 & 1 & 1 & 0 & -1 & 0 & 0 & 0 & $(a,b|xy)$&\\[2pt]
0 & -1 & 0 & 0 & 0 & 0 & 0 & 0 & $(x,y|a)$&\\[2pt]
0 & 0 & -1 & 0 & 0 & 0 & 0 & 0 & $(x,y|b)$&\\[2pt]
0 & 0 & 0 & 1 & 0 & -1 & 0 & 0 & $(a,x|cy)$&  $(a,y|cx)$\\[2pt]
0 & 0 & 0 & 0 & 1 & 0 & 0 & -1 & $(c,x|aby)$&  $(c,y|abx)$\\[2pt]
0 & 0 & 0 & 0 & 0 & 1 & 0 & -1 & $(b,x|acy)$&  $(b,y|acx)$\\[2pt]
0 & 0 & 0 & 0 & 0 & 0 & 0 & 1 & $(x|abcy)$& $(y|abcx)$\\[2pt]
\hline
\end{tabular}\end{center}\vskip -5pt
\caption{A submatrix of $\mathcal M_{abcxy}$}\label{table:3-facets}
\end{table}
Table \ref{table:3-facets} shows some rows of $\mathcal M'$
with the corresponding facet equations (one or two). The matrix $\mathcal
M'$ can be constructed by hand, or by some interpretative computer program.
The next step is to extract the extremal non-negative linear combinations of
the rows which give zero sums for all eight columns. This can be done, e.g.,
by the freely available software packages {\sc Porta} \cite{porta}. The
result is 154 extremal non-negative linear combinations.
One of them is the combination taking all but the
first and last row from Table \ref{table:3-facets} once, and taking the last
row twice. Eight of the corresponding 32 facet combinations give the
inequalities in (\ref{eq:amalgamproof}), which, after rearranging the terms,
give the inequalities in (\ref{eq:amalgam3}).
The other $3\cdot 8$ combinations,
when one takes $(c,x|y)$ instead of $(c,y|x)$, or $(a,y|cx)$ instead of
$(a,x|cy)$, or both, yield supporting hyperplanes to the projected cone, but
not facets as they are consequences of the basic (Shannon) inequalities for
$abcx$ and $abcy$. In other words, these hyperplanes do not cut into the cones
$\Gamma_{abcx}$ and $\Gamma_{abcy}$.

All other bounding hyperplanes (inequalities) resulting from the remaining
153 extremal combinations were checked by an interpretative computer
program whether they are really facets of the projection. This search
resulted in the statement of the Theorem.
\end{proof}

\begin{corollary}\label{corr:3-adhesive}
Polymatroids $f_x$ and $f_y$ on the ground sets $abcx$ and $abcy$ are
adhesive if and only if the following four inequalities and their permutations
hold:
\begin{align}\label{eq:adhesive3}
  & ~~~ f_x(a,x|c)+f_x(a,b|x)+f_y(a,b|y)+f_y(c,y) +{}\\
  &+\twoline{f_x(b,x|ac)}{f_y(b,y|ac)} + \twoline{f_x(c,x|ab)}{f_y(c,y|ab)} \ge
f_x(a,b).\nonumber
\end{align}
\end{corollary}
\begin{proof}
By Theorem \ref{thm:amalgam-Xy}, $f_x$ and $f_y$ are adhesive if and only if
$f_x\down x$ and $f_y\down y$ have an amalgam. All terms in
(\ref{eq:amalgam3})
are the same for $f_x$ and $f_x\down x$ ($f_y$ and
$f_y\down y$) except for $(f_x\down x)(x|abc)=0$ and 
$(f_y\down y)(y|abc)=0$.
\end{proof}

\begin{corollary}\label{corr:5ineq}
The following are four five-variable non-Shannon information inequalities, that
is, they hold in every entropic polymatroid on  at least five elements:
\begin{align*}
  & ~~~ (a,x|c)+(a,b|x)+(a,b|y)+(c,y) +{}\\
  &+\twoline{(b,x|ac)}{(b,y|ac)} + \twoline{(c,x|ab)}{(c,y|ab)} \ge
(a,b).
\end{align*}
\end{corollary}
\begin{proof}
As observed in \cite{fmadhe}, restrictions of an entropic polymatroid are
adhesive, consequently the inequalities (\ref{eq:adhesive3}) in Corollary
\ref{corr:3-adhesive} must hold.
\end{proof}

\subsection{Examples}\label{subsec:examples}

\begin{example}\label{example:ex1}
There are linearly representable polymatroids $f_x$ and $f_y$ on $abcx$ and
$abcy$ which have an amalgam but are not adhesive.
\end{example}
\begin{proof}
Polymatroids $f_x$ and $f_y$ will be extensions of the uniform polymatroid
$$
   f(A) = \left\{\begin{array}{ll}
            4 & \mbox{ if $|A|=1$},\\
            6 & \mbox{ otherwise},
          \end{array}\right. ~~~~ A\subseteq \{abc\}.
$$
Clearly $f(i,j)=f(i,j|k)=2$ for
all distinct $i$, $j$, $k$. The excess functions defining $f_x$ and $f_y$ are
$$
 e_x(A)=\left\{\begin{array}{ll} 
            3 & \mbox{ if $A=\emptyset$, }\\
            1 & \mbox{ otherwise},
        \end{array}\right.
~~ \mbox{ and } ~~
 e_y(A)=\left\{\begin{array}{ll} 
            3 & \mbox{ if $A=\emptyset$ or $A=\{c\}$,}\\
            1 & \mbox{ otherwise}.
        \end{array}\right.
$$
\newcommand\mf[1]{\mathbf{#1}}%
By Examples \ref{example:excess-1} and \ref{example:excess-2} both
$f_x$ and $f_y$ are polymatroids. They are {\em not} adhesive, as all terms on the
left hand side of (\ref{eq:adhesive3}) are zero, while $f_x(a,b)=f(a,b)=2$.
To show that they have an amalgam, one can check that all conditions of
Theorem \ref{thm:amalgam3} hold. The polymatroid $f_{xy}$ specified in 
Table \ref{table:ex1} gives such an extension explicitly.
\begin{table}[!htb]\begin{center}\begin{tabular}{|ccc@{\qquad}ccc@{\qquad}ccc@{\qquad}ccc|}
\hline\multicolumn{3}{|c@{\qquad}}{\rule{0pt}{12pt}$A$}&
      \multicolumn{3}{@{}c@{\qquad}}{$Ax$ }&
      \multicolumn{3}{@{}c@{\qquad}}{$Ay$ }&
      \multicolumn{3}{@{}c|}{$Axy$ }\\[2pt]
\hline\rule{0pt}{13pt}
    & 6 &   &    & 7 &   &    & 7 &   &    & 7 &     \\[2pt]
~ 6 & 6 & 6 &  7 & 7 & 7 &  7 & 7 & 7 &  7 & 7 & 7~  \\[2pt]
~ 4 & 4 & 4 &  5 & 5 & 5 &  5 & 5 & 7 &  6 & 6 & 7~  \\[2pt]
    & 0 &   &    & 3 &   &    & 3 &   &    & 5 &     \\[2pt]
\hline
\end{tabular}\end{center}\vskip -5pt
\caption{The polymatroid $f_{xy}$ for $A\subseteq \{abc\}$}\label{table:ex1}
\end{table}
The four groups contain the values for the subsets indicated at the top line
where $A$ runs over all subsets of $abc$. The values are arranged
in four lines (from bottom to top) for $A=\emptyset$, one-element
subsets $a$, $b$, $c$, two-element subsets $ab$, $ac$, $bc$, and $abc$ at
the top.

Finally, the polymatroids $f_x$ and $f_y$ are linearly representable over
any field. Choose seven independent vectors $\mf s_1$, $\mf s_2$, $\mf u_1$,
$\mf u_2$, $\mf v_1$, $\mf v_2$, and $\mf r$. Subspaces assigned to the
ground elements are are the ones spanned by the vectors listed below:
$$\begin{array}{rl}
\multicolumn{2}{c}{\mbox{$abcx$ is linear}}\\[3pt]
a: & \mf s_1,\mf s_2, \mf u_1,\mf u_2 \\[2pt]
b: & \mf s_1, \mf s_2, \mf v_1, \mf v_2\\[2pt]
c: & \mf s_1,\mf s_2,\mf u_1+\mf v_1, \mf u_2+\mf v_2 \\[2pt]
x: & \mf s_1,\mf s_2,\mf r
\end{array}
\qquad
\begin{array}{rl}
\multicolumn{2}{c}{\mbox{$abcy$ is linear}}\\[3pt]
a: & \mf s_1,\mf s_2, \mf u_1,\mf u_2 \\[2pt]
b: & \mf s_1, \mf s_2, \mf v_1, \mf v_2\\[2pt]
c: & \mf u_1,\mf u_2,\mf v_1, \mf v_2 \\[2pt]
y: & \mf s_1,\mf s_2,\mf r
\end{array}
$$
It is easy to check that all generated subspaces have the right dimension.
Note that while the dimensions of the subspaces corresponding to subsets of 
$abc$ are the same, the subspace arrangements are {\em not} isomorphic.
\end{proof}

It is easy to check that $f_x\down x$ and $f_y\down y$ are also linearly
representable. As $f_x$ and $f_y$ are not adhesive, according to Theorem
\ref{thm:amalgam-Xy}, $f_x\down x$ and $f_y\down y$ have no amalgam.

\medskip

Theorem \ref{thm:amalgam3} can be used to characterizes 1-sticky polymatroids on
three-element sets. The following examples show some particular cases.

\begin{example}\label{example:nonsticky-1}
Let $f$ be a polymatroid on $\{abc\}$. If $f(a,b)$,
$(a,b|c)$ are positive, $(a,b)\le (a,c)$, $(a,b)\le (b,c)$,
then $f$ is not 1-sticky.
\end{example}
\begin{proof}
We specify two extensions $f_x$ and $f_y$ so that one of the
inequalities in Theorem \ref{thm:amalgam3} fails.
Let $t=(a,b)>0$, and $u=\min\{ (a,b),~ (a,b|c) \}>0$.
Define the excess functions $e_x$, $e_y$ by
$$
e_x(A) = \left\{
\begin{array}{rl}
    t & \mbox{ if $A=\emptyset$,}\\
    0 & \mbox{ otherwise},
\end{array}\right.
~~~\mbox{ and } ~~~
e_y(A)=\left\{
\begin{array}{rl}
    u & \mbox{ if $A=\emptyset$ or $A=\{c\}$,} \\
    0 & \mbox{ otherwise}.
\end{array}\right.
$$
According to Examples \ref{example:excess-1} and \ref{example:excess-2},
$f_x$ and $f_y$ are polymatroids. In this case
$f_x(a,x|c)=f_x(b,x|ac)=f_x(c,x|ab)=0$, $f_x(x|abc)=0$, and
$f_x(a,b|x)=f(a,b)-t=0$, see the remark following Example
\ref{example:excess-1}. Similarly, we have 
$f_y(a,b|y)=f(a,b)-u=t-u$, $f_y(c,y)=0$, thus the
left hand side of the top line in (\ref{eq:amalgam3}) is
\begin{align*}
  & f_x(a,x|c)+(f_x(a,b|x)+f_y(a,b|y)+f_y(c,y)+{}\\
  +& f_x(b,x|ac)+f_x(c,x|ab)+2f_x(x|abc) = t-u,
\end{align*}
while the right hand side is $f(a,b)=t$. Thus no amalgam of $f_x$ and $f_y$
exists.
\end{proof}

\begin{example}\label{example:is-sticky}
Suppose $(a|bc)=(b|ac)=(c|ab)=0$, and at least one of $(a,b|c)$,
$(a,c|b)$, $(b,c|a)$ is zero. Then $f$ is 1-sticky.
\end{example}
\begin{proof}
Let $f_x$ and $f_y$ be two extensions of $f$. Our goal is to show that all
instances of the inequalities in Theorem \ref{thm:amalgam3} hold. From the
assumptions it follows that for $|A|\ge 2$ we have $f(A)=f(abc)=t$; moreover
at least one of $f(a)$, $f(b)$, $f(c)$ also equals $t$. Suppose
$f_x$ and $f_y$ are specified by the excess functions $e_x$ and
$e_y$. By Proposition \ref{prop:tight-XY} we can assume that $f_x$ is tight
on $x$ and $f_y$ is tight on $y$, which gives $e_x(M)=e_y(M)=0$, where
$M=\{abc\}$. In our case
$f(i|M{-}i)=0$, thus we must also have $e_x(i|M{-}i)=e_y(i|M{-}i)=0$, thus
$e_x(A)=e_y(A)=0$ for all two-element subsets of $M$. This means
$$
   f_x(i,x|M{-}i) = f_x(x|M)= 0, ~~~ f_y(i,y|M{-}i)=f_y(y|M)=0,
$$
thus all terms in the second line of (\ref{eq:amalgam3}) are zero.
Consequently we only need to
show that
$$
  f_x(a,x|c)+f_x(a,b|x)+f_y(a,b|y)+f_y(c,y)\ge f(a,b),
$$
which rewrites to
\begin{equation}\label{eq:is-sticky}
  f(a,b)+e_x(a)+e_x(b)+e_x(c)-e(x) + e_y(a)+e_y(b)-e_y(c)\ge 0.
\end{equation}
The condition that one of $f(a)$, $f(b)$, $f(c)$ equals $t$ was not used
yet.
If $f(c)=t$, then $e_x(c)=e_y(c)=0$, and then (\ref{eq:is-sticky})
follows from
$$
    f(a,b)+e_x(a,b)+e_y(a)+e_y(b) \ge 0,
$$
which holds by Lemma \ref{lemma:excess}, Condition 3. When $f(a)=t$ (or,
symmetrically, $f(b)=t$), then $e_x(a)=e_y(a)=0$,
$f(a,b)=f(b)=f(b,c)+f(a,b|c)$, and (\ref{eq:is-sticky}) rewrites to
$$
   f(b,c)+e_x(b,c)+f(a,b|c)+e_y(a,b|c)+e_y(b)\ge 0,
$$
which, again, holds by Lemma \ref{lemma:excess}.
\end{proof}

\section{Two-element extensions}\label{sec:sticky-3}

Using similar techniques necessary and sufficient conditions for the
existence of an amalgam of polymatroids on larger sets can be obtained.
Theorem \ref{thm:sticky-2-1} is
such an example. It is a consequence of \cite[Remark 6]{fmadhe} and Theorem
\ref{thm:amalgam-Xy}; we sketch a direct proof. The result is used to
get a characterization of 2-sticky polymatroids on two-element sets.

\begin{theorem}\label{thm:sticky-2-1}
Polymatroids $f_X$ and $f_y$ on $abx_1x_2$ and $aby$, respectively, have an amalgam if and
only if the following two inequalities and all of their permutations (permuting
$a$ and $b$, and $x_1$ and $x_2$) hold choosing either the top or the bottom
line from the list in curly brackets:
$$
   \twoline{f_X[a,b,x_1,x_2]}{f_X[a,x_1,b,x_2]}
   + f_y(y,a|b)+f_y(y,b|a)+f_y(a,b|y)
 +3f_y(y|ab)\ge 0.
$$
\end{theorem}
\begin{proof}
After expanding and rearranging the above inequalities are equivalent to
\begin{align*}
&~~~\twoline{(x_2,y|b)+(x_1,x_2|y)+(a,b|x_2y)+(a,y|x_1x_2)}
            {(x_2,b|y)+(x_2,y|x_1)+(a,y|bx_2)+(a,x_1|x_2y)}+{}\\
&+ (x_1,y|a)+(x_2,y|a)+(x_1,y|b)+(a,b|x_1y)+{}\\
&+(y|abx_1)+(y|abx_2)+(y|ax_1x_2) \ge 0,
\end{align*}
thus if $f_X$ and $f_y$ have an amalgam, then the expressions must be
non-negative.

The sufficiency can be checked similarly as in Theorem \ref{thm:amalgam3} by
computing the facets of the projection of the cone $\Gamma_{\{abx_1x_2y\}}$ 
\begin{table}[htb]
\begin{center}\begin{tabular}{|ccc|ccc|ccc|ccc|l|}
\hline
\rule{0pt}{11pt}$x_1y$&$x_2y$&$x_1x_2y$&\multicolumn{3}{|c}{$a$}&\multicolumn{3}{|c}{$b$}&\multicolumn{3}{|c|}{$ab$}&\\
\hline
\rule{0pt}{12pt}
  1&0&0& -1&0&0& 0&0&0& 0&0&0  &$(a,x_1|y)$, $(a,y|x_1)$\\[2pt]
-1&0&0& 1&0&0& 1&0&0& -1&0&0   &$(a,b|x_1y)$\\[2pt]
 1&1&-1&0&0&0& 0&0&0&  0&0&0   &$(x_1,x_2|y)$\\[2pt]
 0&-1&0& 0&1&0& 0&1&0& 0&-1&0  &$(a,b|x_2y)$\\[2pt]
 0&0&1&  0&0&-1&0&0&0& 0&0&0   &$(a,y|x_1x_2)$\\[2pt]
 0&0&0&-1&0&0& 0&0&0& 0&0&0    &$(x_1,y|a)$\\[2pt]
 0&0&0& 0&-1&0&0&0&0& 0&0&0    &$(x_2,y|a)$\\[2pt]
 0&0&0& 0&0&1 &0&0&0& 0&0&-1   &$(b,y|ax_1x_2)$\\[2pt]
 0&0&0& 0&0&0&-1&0&0& 0&0&0    &$(x_1,y|b)$\\[2pt]
 0&0&0& 0&0&0& 0&-1&0&0&0&0    &$(x_2,y|b)$\\[2pt]
 0&0&0& 0&0&0& 0&0&0& 1&0&-1   &$(x_2,y|abx_1)$\\[2pt]
 0&0&0& 0&0&0& 0&0&0& 0&1&-1   &$(x_1,y|abx_2)$\\[2pt]
 0&0&0& 0&0&0& 0&0&0& 0&0&1    &$(y|abx_1x_2)$\\[2pt]
\hline
\end{tabular}\end{center}\vskip -5pt
\caption{A submatrix of $\mathcal M_{\{abuxy\}}$}\label{table:proj-facets}
\end{table}
to the coordinates which are subsets of $abx_1x_2$ and $aby$. There are 12
dropped coordinates: $x_1y$, $x_2y$, $x_1x_2y$, \dots, $abx_1y$, $abx_2y$,
$abx_1x_2y$. Restricting the matrix $\mathcal M$ describing the facets of 
the cone 
$\Gamma_{\{abx_1x_2y\}}$ to these columns, one gets the submatrix $\mathcal
M'$ with 48 rows and 12 columns. Some of the rows are shown in Table \ref{table:proj-facets}.
Software Porta \cite{porta} found 6938 extremal non-negative linear
combinations giving zero sums for the 12 projected variables. 
One facet of the projection is generated by the linear combination taking
all but the first and last rows from Table \ref{table:proj-facets} once, 
and the last row three times. As in case of Theorem \ref{thm:amalgam3} all
extremal combinations were expanded to bounding hyperplanes and checked
whether it was a facet of the projection.
This search confirmed the claim.
\end{proof}

\begin{theorem}\label{thm:nonsticky-2}
The polymatroid on the two-element set $ab$ is 2-sticky if and only if
one of the following cases hold: $(a,b)=0$ (it is modular); $(a|b)=0$, or
$(b|a)=0$ (one of them determines the other).
\end{theorem}
\begin{proof}
First we show that these polymatroids are 2-2-sticky. Modular polymatroids are sticky
without any restriction, so suppose, e.g., that $f(a|b)=0$. Let $f_X$ be an
extension on $abx_1x_2$.
All six
Ingleton expressions for $f_X$ are non-negative using the following
equalities and their symmetric versions:
\begin{align*}
   [a,b,x_1,x_2] &+(a|b)=(a,x_1|b)+(a,x_2|b)+(x_1,x_2|a)+(a|x_1x_2);\\
   [a,x_1,b,x_2] &+(a|b)=(a,x_1|b)+(b,x_2|a)+(a,x_2|x_1)+(a|bx_2);\\
   [b,x_1,a,x_2]
   &+(a|b)=(a,x_1|b)+(a,x_2|b)+(a,x_1|x_1)+(b,x_1|ax_2)+(a|bx_1x_2);\\
   [x_1,x_2,a,b]
   &+(a|b)=(a,x_1|b)+(a,x_2|x_1)+(a,b|x_2)+(x_1,x_2|ab)+(a|bx_1x_2).
\end{align*}
It means that $f_X$ is linear, and the same is true for $f_Y$. As in the
proof of Theorem \ref{thm:1-1-sticky}, using Proposition \ref{prop:pm-approx}
we may assume that $f_X$ and $f_Y$ are linearly representable over the same
field, and the
dimensions of the subspaces corresponding to the common subsets $a$, $b$
and $ab$ are the same in both representations. 
Choose maximal independent set of vectors in both representations
which span these subspaces in an equivalent way. Extend this set to be a
base in both representations. Glue the 
two vector spaces together along the equivalent set of base vectors. 
This gives an amalgam (even an adhesive extension) of $f_X$ and $f_Y$, as
required.

In the other direction first we show that given $f$ with $f(a,b)>0$, $f(a|b)>0$,
$f(b|a)>0$, it can be extended to a polymatroid $f_X$ on $abx_1x_2$
so that $f_X[a,b,x_1,x_2]<0$. For the construction we recall the {\em 
natural coordinates} of polymatroids on four elements from \cite{entreg}. 
This coordinate system has the additional advantage
that points with natural coordinates in the non-negative orthant $\R_{\ge 0}^{15}$
are polymatroids.
Let us recall these coordinates below:
\begin{align*}
   & -[a,b,x_1,x_2], \\
   & (a,b|x_1),~ (a,b|x_2), ~ (a,x_1|b), ~ (b,x_1|a), ~ (a,x_2|b), ~ (b,x_2|a),\\
   & (x_1,x_2|a), ~ (x_1,x_2|b), ~ (x_1,x_2), ~(a,b|x_1x_2), \\
   & (a|bx_1x_2), ~ (b|ax_1x_2), ~ (x_1|abx_2), (x_2|abx_1).
\end{align*}
From these coordinates the values $f_X(a)$, $f_X(b)$, and $f_X(ab)$ can be 
expressed as
follows, where only the coefficients of the coordinates in the above order
are shown:
\def\sb#1{\mathbf #1}
\begin{align*}
   f_X(a) &= (\sb2,~ 1,1,1,0,1,0,~ \sb1,\sb0,\sb1,1,~ 1,0,0,0)\\
   f_X(b) &= (\sb2,~ 1,1,0,1,0,1,~ \sb0,\sb1,\sb1,1,~ 0,1,0,0)\\
  f_X(ab) &= (\sb3,~ 1,1,1,1,1,1,~ \sb1,\sb1,\sb1,2,~ 1,1,0,0).
\end{align*}
Choose the coordinates first and from eighth to tenth (typeset in bold)
to have the positive values $\eps$, $f(a|b)-\eps$, $f(b|a)-\eps$, and 
$f(a,b)-\eps$, respectively, for some small enough $\eps$; set all other
coordinates to zero. With this choice $f_X$ will be a polymatroid which
extends the one given on $ab$ as, e.g.,
$f_X(a)=2\eps+f(a,b)-\eps+f(a|b)-\eps=f(a)$,
moreover the Ingleton value $f_X[a,b,x_1,x_2]$, as given by the first coordinate, 
is $-\eps$, which is negative.

Define the other extension $f_y$ by the excess function
$$
    e_y(A) = \left\{\begin{array}{ll}
               f(a,b) & \mbox{ if $A=\emptyset$,}\\
               0     & \mbox{ otherwise}.
             \end{array}\right.
$$
By the remark after Example \ref{example:excess-1}, $f_y$ is a polymatroid
and $f_y(a,b|y)=f_y(a,y|b)=f_y(b,y|a)=0$ as well as $f_y(y|ab)=0$.
According to Theorem \ref{thm:sticky-2-1} if $f_X$ and $f_y$ have an amalgam, 
they must satisfy
$$
   f_X[a,b,x_1,x_2] + f_y(a,b|y)+f_y(a,y|b)+f_y(b,y|a)+3f_y(y|ab)\ge 0.
$$
This value, however, is $-\eps<0$, which proves the theorem.
\end{proof}

\section*{Acknowledgment}

I am indebted to the reviewers for their valuable work.
The careful and thorough reports helped to improve the
presentation, clarify terminology, and made the paper
more accessible to the interested readers.

The author would like to thank the generous support of 
the Institute of Information Theory and Automation of the CAS, Prague.
The research reported in this paper was supported by GACR project
number 19-04579S, and partially by the Lend\"ulet program of the HAS.

\makecontacts

\end{document}